\documentclass[11pt]{article}

\usepackage[letterpaper,margin=0.99in]{geometry}
\usepackage{amsmath}
\usepackage{amssymb,mathtools}
\usepackage{amsthm}
\usepackage{graphicx}
\usepackage[framemethod=tikz]{mdframed}
\usepackage[disable,backgroundcolor=gray!20, linecolor=gray, textsize=scriptsize]{todonotes}
\usepackage{enumerate}
\usepackage{cite}
\usepackage{caption}
\usepackage{tikz}
\usepackage{subcaption}

\newtheorem{theorem}{Theorem}[section]

\newtheorem{remark}[theorem]{Remark}
\newtheorem{corollary}[theorem]{Corollary}
\newtheorem{definition}[theorem]{Definition}
\newtheorem{lemma}[theorem]{Lemma}

\newtheorem{claim}[theorem]{Claim}

\newcommand{\E}{\mathbb{E}}

\usepackage{varioref}
\definecolor{darkgreen}{rgb}{0,0.5,0}
\usepackage{hyperref}
\hypersetup{
    unicode=false,          
    colorlinks=true,        
    linkcolor=red,          
    citecolor=darkgreen,        
    filecolor=magenta,      
    urlcolor=cyan           
}
\usepackage[capitalize, nameinlink]{cleveref}
\Crefname{lemma}{Lemma}{Lemmas}
\Crefname{corollary}{Corollary}{Corollaries}
\Crefname{remark}{Remark}{Remarks}
\Crefname{theorem}{Theorem}{Theorems}

\title{Improved Network Decompositions using Small Messages \\ with Applications on MIS, Neighborhood Covers, and Beyond}
\author{
  Mohsen Ghaffari\\
  ETH Zurich \\
  ghaffari@inf.ethz.ch
  \and
  Julian Portmann \\
  ETH Zurich\\
  pjulian@ethz.ch
}
\date{}

\begin{document}
\maketitle

\begin{abstract}
  \normalsize
  Network decompositions, as introduced by Awerbuch, Luby, Goldberg, and Plotkin [FOCS'89], are one of the key algorithmic tools in distributed graph algorithms.
  We present an improved deterministic distributed algorithm for constructing network decompositions of power graphs using small messages, which improves upon the algorithm of Ghaffari and Kuhn [DISC'18].
  In addition, we provide a randomized distributed network decomposition algorithm, based on our deterministic algorithm, with failure probability exponentially small in the input size that works with small messages as well.
  Compared to the previous algorithm of Elkin and Neiman [PODC'16], our algorithm achieves a better success probability at the expense of its round complexity, while giving a network decomposition of the same quality.
  As a consequence of the randomized algorithm for network decomposition, we get a faster randomized algorithm for computing a Maximal Independent Set, improving on a result of Ghaffari [SODA'19].
  Other implications of our improved deterministic network decomposition algorithm are:
  a faster deterministic distributed algorithms for constructing spanners and approximations of distributed set cover, improving results of Ghaffari, and Kuhn [DISC'18] and Deurer, Kuhn, and Maus [PODC'19];
  and faster a deterministic distributed algorithm for constructing neighborhood covers, resolving an open question of Elkin [SODA'04].

\end{abstract}
\thispagestyle{empty}
\newpage
\setcounter{page}{1}

\section{Introduction}
\label{sec:intro}

We present an improved deterministic distributed algorithm for constructing network decompositions of power graphs using small messages, as well as some improvements for other problems including randomized construction of maximal independent set, and deterministic construction of sparse neighborhood covers, spanners and dominating set approximation.

After introducing our model of computation, we recall the concept of network decompositions in \Cref{sec:intronetdecomp} as well as a brief summary of all known distributed constructions.
In \Cref{sec:results} we present our results and in \Cref{sec:overview} we outline our methods and explain how they depart from previous approaches.

\paragraph{Model:} Throughout, we work with the {\sffamily CONGEST} model of distributed computing \cite{peleg2000distributed}:
The communication network is abstracted as an $n$-node graph $G = (V,E)$.
We use $\Delta$ to denote the maximum degree of $G$.
There is one processor on each node of the network, which initially knows only its $O(\log n)$-bit identifier.
Per round of synchronous communication, every node can send one $O(\log n)$-bit message to each neighbor.
Note that this is enough to describe constantly many elements of the network, i.e. vertices or edges.
A closely related variant is the {\sffamily LOCAL} model\cite{linial1987distributive}, where we impose no restriction on the size of messages.

\subsection{Network Decompositions}
\label{sec:intronetdecomp}
\emph{Network decompositions} were introduced by Awerbuch et al.\cite{awerbuch1989network}, and since then, they have turned out to be one of the key algorithmic tools in distributed algorithms for graph problems.
For a given graph $G=(V, E)$, a $(c,d)$ network decomposition of it is defined as a partition of $V$ into \emph{blocks} $V_1, \dots, V_c$ such that each connected component of the subgraphs $G[V_i]$ has diameter at most $d$.
The connected components of each block are usually called \emph{clusters}.
This notion of network decomposition is sometimes also referred to as \emph{strong diameter} network decomposition, as we consider the diameter with respect to distances in the induced subgraphs.
This is as opposed to \emph{weak diameter} network decompositions, where distances are with respect to the base graph.
Intuitively, network decompositions allow us to process graph problems in $c$ sequential stages, where in each stage we process one block, a graph that is made of low-diameter components (diameter $d$).
This low-diameter simplifies the task as it opens the road for collecting either the entire topology, in the {\sffamily LOCAL} model, or at least some coordination messages, in the {\sffamily CONGEST} model.
The key point is that the problems in different components of one block can be processed independently, as they have distance at least $1$.

In many applications of network decompositions, instead of asking for the clusters to have distance at least $1$, we need them to have a larger distance, at least $k$ hops for some parameter $k\geq 2$.
This is crucial for applications where the problem is such that the answer in one node can impact nodes beyond its neighbors.
Thus, a natural extension of network decomposition is the following:
a $k$-hop separated network decomposition or decomposition of $G^k$ requires that any two nodes $u, v$ from different clusters of the same color are at distance more than $k$ in $G$.
We note that clusters do not have to be connected in $G$, which means that it is a \emph{weak diameter} decomposition of $G$.

While the authors of \cite{awerbuch1989network} used network decompositions to solve symmetry breaking problems, such as maximal independent set or $(\Delta+1)$-vertex coloring, various other applications were discovered later.
Examples in the {\sffamily LOCAL} model include the computation of sparse spanners and linear-size skeletons by Dubhashi et al.\cite{dubhashi2005fast} or distributed approximation algorithms for the graph coloring and minimum dominating set problems by Barenboim et al.\cite{barenboim2012locality, barenboim2018fast}.
For the {\sffamily CONGEST} model, Ghaffari and Kuhn\cite{ghaffari2018derandomizing} showed that $k$-hop separated network decompositions can be used for computing spanners and approximating minimum dominating set.

\paragraph{State of the Art---Deterministic Constructions:}
There are four known deterministic distributed constructions of network decompositions, successively improving either quantitatively or qualitatively \cite{awerbuch1989network, panconesi1992improved, ghaffari2018derandomizing, ghaffari2019distributed}.
Awerbuch et al. \cite{awerbuch1989network} provided an algorithm for computing $(2^{O(\sqrt{\log n \log \log n})}, 2^{O(\sqrt{\log n \log \log n})})$ network decompositions of an $n$ node graph $G$ in $2^{O(\sqrt{\log n \log \log n})}$ rounds, which works in the {\sffamily CONGEST} model.
Subsequently, this was improved by Panconesi and Srinivasan\cite{panconesi1992improved} showing that all $2^{O(\sqrt{\log n \log \log n})}$ terms could be replaced by $2^{O(\sqrt{\log n})}$.
However, their algorithm requires large messages.

For network decompositions with higher levels of separation, Ghaffari and Kuhn\cite{ghaffari2018derandomizing} gave a $k \cdot 2^{O(\sqrt{\log n \log \log n})}$ round {\sffamily CONGEST}-model algorithm for computing a $(2^{O(\sqrt{\log n \log \log n})}, 2^{O(\sqrt{\log n \log \log n})})$ network decomposition of $G^k$, which works with small messages.
Note that extending network decomposition algorithms to compute a decomposition of $G^k$ is trivial in the {\sffamily LOCAL} model:
As nodes can send messages of arbitrary size, communication on $G^k$ can be simulate in $k$ rounds of communication on $G$.
Thus, with a $k$ factor overhead in the round complexity (and a $k$ factor increase in the diameter with respect to distances in $G$), we can use any {\sffamily LOCAL}-model network decomposition algorithm to also compute $k$-hop separated decompositions.

Recently, Ghaffari\cite{ghaffari2019distributed} showed that a $(2^{O(\sqrt{\log n})}, 2^{O(\sqrt{\log n})})$ network decomposition can also be computed in $2^{O(\sqrt{\log n})}$ rounds in the {\sffamily CONGEST} model.
However, his construction cannot extend to $G^k$, which is one of the issues we address in this paper.
In contrast to all previous approaches, this algorithm has the useful property, that it can handle large identifiers.
This means that the length of identifiers does not influence the parameters of the resulting network decomposition.


\paragraph{State of the Art---Randomized Constructions:} For randomized algorithms there are stronger results:
Linial and Saks\cite{linial1993low} showed that $(O(\log n), O(\log n))$ network decompositions exist and gave a distributed algorithm, which finds a $(O(\log n), O(\log n))$ network decomposition in $O(\log^2 n)$ rounds, with high probability\footnote{As usual, we use the phrase \emph{with high probability} to denote that an event holds with probability at least $1 - n^{-c}$ for any constant $c$, where $c$ may influence other constants.} (w.h.p).
The construction of Linial and Saks\cite{linial1993low} only guarantees that clusters have weak diameter $O(\log n)$.
More recently, Elkin and Neiman\cite{elkin2016distributed} provided a randomized distributed algorithm that computes strong diameter $(O(\log n), O(\log n))$ network decomposition in $O(\log^2 n)$, w.h.p, and also works in the {\sffamily CONGEST} model.
Both of these algorithms can be easily extended to produce a $(O(\log n), O(k \log n))$ decomposition of $G^k$ in $O(k \log^2 n)$ rounds without requiring larger messages.

We remark that the fact that these algorithm succeed with probability $1 - 1/\text{poly}(n)$ prevents them from being directly used in our randomized MIS algorithm.
This is because after the shattering, only components of size $N \ll n$ remain, which means that the algorithms only succeed with probability $1 - 1/\textrm{poly}(N)$ in computing a $(O(\log N), O(\log N))$ network decomposition.

\subsection{Our Results}
\label{sec:results}
We present a deterministic distributed {\sffamily CONGEST}-model algorithm for computing network decompositions of $G^k$:
\begin{theorem}
  \label{thm:netdecompintro}
  There is a deterministic distributed algorithm that in any $N$-node network $G$, which has $S$-bit identifiers and supports $O(S)$-bit messages for some arbitrary $S$,
  computes a $(g(N), g(N))$ network decomposition of $G^k$ in $k g(N) \cdot \log^* S$ rounds, for any $k$, and $g(N) = 2^{O(\sqrt{\log N})}$.
\end{theorem}
We highlight the following three properties, whose combination is new to our algorithm and is crucial for our applications in the next subsections:
(A) it is able to compute a network decomposition of $G^k$ in the {\sffamily CONGEST} model, (B) it can handle large identifiers, and (C) its bounds are as good as a simulation of the algorithm of \cite{panconesi1992improved} on $G^k$ in the {\sffamily LOCAL} model.
More precisely, property (B) says that the size of identifiers only affects the round complexity but not the quality of the computed network decomposition.

In particular for our application to MIS, property (B) is crucial.
The lack of this ability to handle large identifiers is what made previous algorithms, such as of Ghaffari and Kuhn\cite{ghaffari2018derandomizing}, not applicable.
We refer to \Cref{sec:overview} for a more in-depth explanation of these issues.

\subsection*{Applications: MIS, Neighborhood Cover, and Beyond}
Network decompositions have a wide range of applications and due to previous work, our new algorithms leads to an improvement for a number of problems.
While some of these results are immediate, for the application to MIS, we also present a randomized algorithm for network decompositions whose failure probability is exponentially small in the input size.

\paragraph{The MIS Problem:}
The Maximal Independent Set Problem (MIS) asks for a set $S$ of nodes, such that no two neighboring nodes are in $S$ and moreover, for each node $v$, either $v$ or at least one of its neighbors is in $S$.
It is one of the most well-studied distributed graph problems.
One reason for its importance is that other fundamental graph problems, such as $(\Delta+1)$-vertex coloring, maximal matching, or 2-approximation of vertex cover reduce to it \cite{linial1987distributive, luby1986simple}.

Luby\cite{luby1986simple} as well as Alon, Babai and Itai\cite{alon1986fast} gave randomized distributed MIS algorithms in the {\sffamily CONGEST} model that have round complexity $O(\log n)$.
The first significant improvement over this run time was due to Barenboim, Elkin, Pettie, and Schneider\cite{barenboim2016locality}, who gave a randomized distributed $O(\log^2 \Delta) + 2^{O(\sqrt{\log \log n})}$ round algorithm.
This bound was then improved by Ghaffari\cite{ghaffari2016improved} to $O(\log \Delta) + 2^{O(\sqrt{\log \log n})}$, which remains the state of the art.
However, both these improvements do not work in the {\sffamily CONGEST} model, as they require messages of up to $\mathrm{poly}(\Delta, \log n)$ bits to gather certain local topologies.
The only improvement upon the algorithms of \cite{luby1986simple, alon1986fast} in the {\sffamily CONGEST} model is due to Ghaffari\cite{ghaffari2019distributed}, who gave a randomized distributed algorithm that runs in $\min\{\log \Delta \cdot 2^{O(\sqrt{\log \log n})}, O(\log \Delta \cdot \log \log n) + 2^{O(\sqrt{\log \log n \cdot \log\log\log n})}\}$ rounds.

We improve this result for all values of $\Delta$ and obtain the following:
\begin{theorem}
  There is a randomized distributed algorithm, with $O(\log n)$-bit messages, that computes an MIS in $O\big(\log \Delta \cdot \sqrt{\log \log n}\big) + 2^{O(\sqrt{\log \log n})}$ rounds, w.h.p.
\end{theorem}
Apart from our improved network decomposition, this result contains a randomized algorithm, that transforms a network decomposition of $G^k$ into a decomposition of $G$ with improved parameters.
This transformation works in the {\sffamily CONGEST} model and succeeds with probability exponential in the input size, which is crucial for its application in solving MIS.
For a more detailed overview, see \Cref{sec:overview}.

\paragraph{Neighborhood Covers and MST:}
Neighborhood covers, as introduced by Awerbuch and Peleg\cite{awerbuch1990sparse} are another form of locality-preserving graph representations and closely related to network decompositions.
A $s$-sparse $k$-neighborhood cover of diameter $d$ is defined as a collection of clusters $C \subseteq V$ such that
(A) for each cluster $C$, we have a rooted spanning tree of $G[C]$ with diameter at most $d$,
(B) each $k$-neighborhood of $G$ is completely contained in some cluster, and
(C) each node of $G$ is in at most $s$ clusters.
Like network decompositions, this form of graph representation has many applications in distributed computing, such as in routing \cite{awerbuch1992routing}, shortest paths \cite{afek1993sparser}, job scheduling and load balancing \cite{awerbuch1992online}, or broadcast and multicast \cite{awerbuch1992efficient}.

Awerbuch and Peleg\cite{awerbuch1990sparse} also gave distributed constructions for sparse neighborhood covers using messages of unbounded size, however they do not extend to the $\mathsf{CONGEST}$ model.
More recently, Ghaffari and Kuhn\cite{ghaffari2018derandomizing} gave the first {\sffamily CONGEST} model algorithm for computing sparse neighborhood covers.
They showed that a $(c, d)$ network decomposition of $G^{2k}$ can be transformed into a $c$-sparse $k$-neighborhood cover of diameter $O(k \cdot d)$ in $O(c(d + k))$ rounds.
Together with their $k \cdot 2^{O(\sqrt{\log n \log \log n})}$ round algorithm for computing a $(2^{O(\sqrt{\log n \log \log n})}, 2^{O(\sqrt{\log n \log \log n})})$ network decomposition of $G^k$, this yields a $k \cdot 2^{O(\sqrt{\log n \log \log n})}$ round algorithm for computing $2^{O(\sqrt{\log n \log \log n})}$-sparse $k$-neighborhood covers of diameter $k \cdot 2^{O(\sqrt{\log n \log \log n})}$.

Using our network decomposition, we improve upon the result of \cite{ghaffari2018derandomizing} and show that all $2^{O(\sqrt{\log n \log \log n})}$ terms can be replaced by $2^{O(\sqrt{\log n})}$.
See \Cref{app:sparsecover}.
\begin{corollary}
  \label[corollary]{cor:sparsecoverintro}
  There is a deterministic distributed algorithm, that for every $k \geq 1$, computes a $2^{O(\sqrt{\log n})}$-sparse $k$-neighborhood cover of diameter $k \cdot 2^{O(\sqrt{\log n})}$ of and $n$-node graph $G$ in $k \cdot 2^{O(\sqrt{\log n})}$ rounds of $\mathsf{CONGEST}$.
\end{corollary}
We also resolve an open question by Elkin\cite{elkin2004faster}, who devised a randomized {\sffamily CONGEST} model algorithm for minimum spanning tree, that runs in $\tilde{O}(\mu(G, \omega) + \sqrt{n})$ rounds, where $\mu(G, \omega)$ is the \emph{MST-radius} of $G$.
The MST-radius $\mu(G, \omega)$ is defined as the smallest value $t$, such that every edge not belonging to the MST of $G$ is the heaviest edge in some cycle of length at most $t$.
However, the only part involving randomness is the construction of neighborhood covers.
The author remarks that the only obstacle towards a deterministic algorithm is that there are no known constructions of sparse neighborhood covers in the {\sffamily CONGEST} model.
Using \Cref{cor:sparsecoverintro}, in \Cref{cor:mst} we get a deterministic distributed {\sffamily CONGEST}-model algorithm for computing MST in $2^{O(\sqrt{\log n})} \cdot (\mu(G, \omega) + \sqrt{n})$ rounds.

\paragraph{Other Problems: Spanners and Dominating Set Approximation}
Due to previous applications of $k$-hop separated network decompositions by Ghaffari and Kuhn\cite{ghaffari2018derandomizing} as well as Deurer, Kuhn, and, Maus\cite{deurer2019deterministic} we obtain the following deterministic {\sffamily CONGEST} model algorithms:
In \Cref{app:spanners}, we review a $2^{O(\sqrt{\log n})}$ round algorithm for computing a $(2k-1)$-stretch spanner with size $O(kn^{1+1/k}\log n)$, and a $O(\log \Delta)$-approximation algorithm for minimum dominating set in $2^{O(\sqrt{\log n})}$ rounds.

\subsection{Method Overview and Comparison with Prior Approaches}
\label{sec:overview}

We first discuss our method for deterministic network decomposition, and then discuss our contribution to the MIS problem.

\paragraph{Network Decomposition:}
The general outline is shared by all known deterministic algorithms for network decomposition \cite{awerbuch1989network, panconesi1992improved, ghaffari2018derandomizing, ghaffari2019distributed}.
This method is often referred to as recursive clustering:
In every step, a number of clusters is merged to form new clusters while some other clusters are added to the output and discarded from the algorithm.
However, there are several challenges in applying this approach in the {\sffamily CONGEST} model, and even more so when aiming to compute a decomposition of $G^k$.

Let us address these challenges in two themes, (A) communication within clusters, and (B) communication between clusters:
Our approach entails the fact that clusters can become disconnected in the base graph $G$, even if they are connected in $G^k$.
While this means that we have more freedom in how we merge clusters, it also requires us to take extra care to allow for intra-cluster communication.
As clusters can now ``overlap'', a single edge of $G$ could be used by many clusters for communication.
We will have a two stage process to ``introduce'' clusters to their neighboring clusters (in $G^k$), which will enable us to bound the ``overlap'' between clusters.
To be more precise we will argue that any edge is used by at most $2^{O(\sqrt{\log n})}$ many clusters for communication.
This allows us to have simultaneous communication in all clusters with just a $2^{O(\sqrt{\log n})}$ factor overhead.

For challenge (B), we would like to simulate communication in $G_C$ on $G$, where $G_C$ is the virtual graph obtained by contracting each cluster into a node and connecting two clusters if they contain nodes that are adjacent in $G^k$.
However, this simulation is not directly possible in the {\sffamily CONGEST} model, as nodes in $G_C$ can have degree much larger than $\Delta$, leading to congestion for communication within clusters.
The solution to this problem will also be the introduction process mentioned above.
Roughly speaking, we will ignore some edges from $G_C$ as well as some vertices of high degree.
This allows for an efficient simulation of communication in $G_C$ on the base graph $G$.

\paragraph{Maximal Independent Set:}
For solving MIS, we follow the outline of the \emph{shattering technique}, first introduced into distributed computing by Barenboim et al.\cite{barenboim2016locality}, and also used for the MIS problem by \cite{ghaffari2016improved, ghaffari2019distributed}.
There are two parts:
In the pre-shattering phase, we solve the problem for a large portion of the graph, leaving only a number of ``small'' connected components.
Then, in the post-shattering phase, we solve the problem on the remaining parts.

In the pre-shattering phase, we use the $O(\log \Delta)$-round algorithm of Ghaffari\cite{ghaffari2016improved}, which works with just single-bit messages.
Afterwards, we are left with ``small'' components.
For now, assume that they have size\footnote{They do not contain $O(\log n)$ nodes, but rather up to $O(\Delta^4 \log n)$ many vertices. However, we will see that we can efficiently cluster them into $O(\log n)$ clusters of diameter only $O(\log \log n)$.} $O(\log n)$.
By computing a network decomposition of each component, we can further simplify the problem:
We go through the color classes, one by one, each time computing an MIS of the new color, that does not conflict with the MIS of the previous colors.
We solve the problem by running $O(\log n)$ independent copies of Ghaffari's $O(\log \Delta)$-round randomized MIS algorithm algorithm \cite{ghaffari2016improved}, all in parallel.
This parallel execution is possible in the {\sffamily CONGEST} model because the algorithm from \cite{ghaffari2016improved} only uses single-bit messages.
With high probability (i.e. at least $1 - 1/\textrm{poly}(n)$), at least one of these independent runs succeeds in computing an MIS.
Using the fact that we are solving the problem in a graph of low diameter, we can efficiently coordinate all nodes to find a successful run.

The main challenge is obtaining a suitable network decomposition:
We could use the network decomposition algorithm from \Cref{thm:netdecompintro} and get an MIS algorithm with round complexity $\log \Delta \cdot 2^{O(\sqrt{\log \log n})}$.
This only matches the previous work of Ghaffari\cite{ghaffari2019distributed}.
Also, randomized algorithms for network decomposition are hard to apply, as we are computing decompositions of graphs that only contain $N = O(\log n)$ nodes.
This means that the success probability of randomized algorithms for network decomposition, such as \cite{linial1993low, elkin2016distributed}, will only be $1 - 1/\textrm{poly}(N) \ll 1 - 1/\textrm{poly}(n)$.

We get around these issues in two steps:
First, we compute a $k$-hop separated network decomposition of each component.
Then, we use this network decomposition to boost the success probability of a randomized network decomposition algorithm, inspired by \cite{elkin2016distributed, miller2013parallel, blelloch2014nearly}.
While Ghaffari\cite{ghaffari2019distributed} used a similar idea to also get an $O(\log \Delta \cdot \log \log n) + 2^{O(\sqrt{\log \log n \cdot \log \log \log n})}$ round algorithm for MIS, we improve upon it in two ways:
First, our network decomposition of $G^k$ has better bounds, and second, we use a randomized process for computing a refined network decomposition.
This randomized process allows us to further reduce the number of colors needed, from $O(\log \log n)$ to $O(\sqrt{\log \log n})$.

\subsection{Mathematical Notation}

For a graph $G = (V,E)$ and two nodes $u, v \in V$, we define $d_G(u,v)$ to be the hop distance between $u$ and $v$.
For a node $v \in V$ and a set $U \subset V$, $dist_G(v,U)$ is the minimum distance between $v$ and any $u \in U$.
For an integer $k \geq 1$ we define the \emph{$k^{th}$ power} $G^k = (V, E')$ of $G$ to be the graph with an edge $\{u, v \} \in E'$ whenever $d_G(u,v) \leq k$.
Given a node $v \in V$, we define $N_{G,k} := \{ u \in V : d_G(u,v) \leq k \}$ to be the \emph{$k$-hop neighborhood} of $v$.

For two integers $\alpha \geq 1$ and $\beta \geq 0$ and a node set $B \subseteq V$, we call $B^* \subseteq B$ a $(\alpha, \beta)$-ruling set of $G$ w.r.t. $B$ if (A) for any two nodes $u, v \in B^*$ we have $d_G(u, v) \geq \alpha$, and (B) $\forall u \in B \setminus B^*$, there is a node $v \in B^*$ such that $d_G(u, v) \leq \beta$.
If $B = V$, $B^*$ is simply called an $(\alpha, \beta)$ ruling set of $G$.

\section{Network Decomposition}
\label{sec:netdecomp}

In this section we describe our algorithm for computing network decompositions of power graphs in the {\sffamily CONGEST} model, as outlined in \Cref{thm:netdecompintro}.
It matches the bounds of the algorithm by Panconesi and Srinivasan\cite{panconesi1992improved}, but improves upon it in three aspects that are crucial to our applications:
our algorithm works in the $\mathsf{CONGEST}$ model, can tolerate large identifiers and is able to produce a network decomposition of $G^k$.
While the first two properties were already achieved by Ghaffari\cite{ghaffari2019distributed}, the third property is new to our approach.

Note that it is trivial to achieve such a decomposition using messages of unbounded size, by just simulating communication in $G^k$ on $G$ (with a $k$ factor overhead).
In the $\mathsf{CONGEST}$ model this idea is not directly applicable and presents two challenges:
(A) how do we deal with clusters being disconnected in the base graph and (B) how do we get around simulating all communication in $G^k$ on $G$?
For the first issue we will bound the number of clusters that are overlapping.
For the second issue we will see that not all communication is necessary.

Before proceeding to the algorithm, we restate \Cref{thm:netdecompintro} in slightly more detail:
\begin{theorem} \label{thm:netdecomp}
  There is a deterministic distributed algorithm that in any $N$-node network $G$, which has $S$-bit identifiers and supports $O(S)$-bit messages for some arbitrary $S$,
  computes a $(g(N), g(N))$ network decomposition of $G^k$ in $k g(N) \cdot \log^* S$ rounds, for any $k$ and $g(N) = 2^{O(\sqrt{\log N})}$.
  Additionally we can simulate one round of communication within clusters of $G^k$ in $k 2^{O(\sqrt{\log N})}$ rounds of communication on $G$.
\end{theorem}

\begin{remark} \label[remark]{rem:netdecomp}
  If we initially have $N$ clusters, each with an $S$-bit center identifier and with radius at most $r$, the algorithm of \Cref{thm:netdecomp} computes a $(g(N), r g(N))$ network decomposition of $G^k$ in $k r g(N) \cdot \log^* S$ rounds.
\end{remark}

\begin{proof}[Proof of \Cref{thm:netdecomp}]
  We first note that the recursive nature of the algorithm makes it directly applicable to use with an initial clustering, as described in \Cref{rem:netdecomp}.

  \paragraph*{Overall Structure:}
  The algorithm consists of phases $i = 1, \dots, \sqrt{\log N}$, each of which runs in $k \cdot 2^{O(\sqrt{\log N})} \cdot \log^* S$ rounds.
  During each phase, the (remaining) vertices are partitioned into vertex-disjoint clusters.
  Each cluster has one center node (which will be the identifier of the cluster), as well as a tree rooted at the center that spans all vertices of this cluster (and potentially also contains vertices of other clusters).
  However we have that any two nodes of the cluster are connected by a path of length at most $k$ in this tree.
  Note that both edges and vertices of $G$ can be included in multiple trees.

  Initially, every node is its own cluster.
  During each phase some clusters join each other to form some new clusters, while the other clusters are colored and removed from the algorithm.
  Let $d = 2^{O(\sqrt{\log N})}$.
  We maintain the following invariants during the $i^{th}$ phase:
  \begin{enumerate}[(A)]
    \item We have at most $n / d^i$ clusters.
    \item The radius of each cluster is at most $h_i = (O(1))^i \leq 2^{O(\sqrt{\log N})}$ in $G^k$, which means its radius is at most $k \cdot 2^{O(\sqrt{\log N})}$ in $G$.
    \item Each edge $e$ is part of at most $i \cdot 13 d^3$ spanning trees, which is always at most $2^{O(\sqrt{\log N})}$.
  \end{enumerate}
  Note that for $i = 0$, these invariants are trivially fulfilled.
  The first invariant ensures that after $\sqrt{\log N}$ phases, there is at most one cluster left, which we can then color with one color and finish the algorithm.
  The second invariant means that cluster radii remain small enough.
  Invariants (B) and (C) together imply that we can perform $2^{O(\sqrt{\log N})} \cdot \log^*S$ iterations of broadcast and convergecast in each cluster within each phase.
  This is because we can simulate a round of communication along an edge in the spanning tree of $G^k$ within $k \cdot 2^{O(\sqrt{\log N})}$ rounds of communication on $G$.

  \paragraph{Informal Outline of each phase:}
  We start out with a set of (old) clusters and will merge some of them into new clusters, while we color the remaining ones and add them to the resulting decomposition.
  In a first step, each cluster $\mathcal{C}$ will try to learn its neighboring clusters.
  If $\mathcal{C}$ has more than $4d^2$ neighbors, we call $\mathcal{C}$ \emph{marked}. 
  If we now consider the graph $\mathcal{G}$ induced by all non-marked clusters, it has maximum degree $\Delta_{\mathcal{G}} \leq 4d^2$.
  This allows us to simulate communication within $\mathcal{G}$ in the underlying network, with about a $4d^2$ overhead in the round complexity (ignoring cluster diameters).
  We will use this fact to find a well-separated set $\mathcal{C}^*$ in $\mathcal{G}$.
  All clusters from $\mathcal{C}^*$, together with the marked clusters will now form the centers of new clusters.
  Then all old clusters that have a neighboring center join this center to form a new cluster.
  Intuitively, as all new cluster centers have high degree, there cannot be a lot of them.
  What we are now left with is a set of low-degree clusters that are not part of any newly formed cluster.
  We can now just color these remaining clusters, using standard coloring techniques, and add them to the final output.
  As the degrees are low, the number of required colors is also low.

  \paragraph*{Building a small in-degree virtual graph $H$:}
  Call two clusters $\mathcal{C}$ and $\mathcal{C}'$ neighboring if they contain vertices $v \in \mathcal{C}$ and $v' \in \mathcal{C}'$ such that $v$ and $v'$ are at distance at most $k$ in $G$.
  This means that $v$ and $v'$ are neighbors in $G^k$.
  Similarly, a node $v$ and a cluster $\mathcal{C}$ are called neighboring if there is some $u \in \mathcal{C}$ such that $u$ is at distance at most $k$ from $v$.

  Now we want every cluster to learn about up to $2d$ many neighboring clusters.
  More precisely, a cluster that has less than $2d$ neighbors should learn about all of its neighbors.
  If it has more than $2d$ neighbors, it learns about some $2d$ of them.
  We can do so in $O(k \cdot d)$ rounds:
  Every node starts a broadcast, sending the identifier of its current cluster to all neighbors.
  Then, over $(2d+1) \cdot (k-1)$ rounds, every node $v$ forwards up to $2d+1$ different such messages about clusters of distance up to $(k-1)$ from $v$.
  This way, if node has at most $2d$ neighboring clusters it learns about all of them and if there are more, it learns about at least $2d$ many, of which it picks some $2d$ many arbitrarily.
  Within clusters, the nodes convergecast at most $2d$ identifiers to the center node.
  This is possible in $O((d + k  \cdot h_i) \cdot 2^{O(\sqrt{\log N})}) = k \cdot 2^{O(\sqrt{\log N})}$ rounds, as invariant (C) states that at most $2^{O(\sqrt{\log N})}$ clusters overlap.
  Thus, a $2^{O(\sqrt{\log n})}$ round overhead is enough to allow all clusters to perform a convergecast at the same time.

  This process creates a directed virtual graph $H$ among the clusters, where an edge $\mathcal{C} \rightarrow \mathcal{C}'$ indicates that the center of $\mathcal{C}'$ received the identifier of $\mathcal{C}$.
  We call a cluster \emph{high-degree} if it has at least $2d$ neighboring clusters, and \emph{low-degree} otherwise, see \Cref{fig:clusters}.
  Notice that a low-degree clusters has all neighboring clusters as incoming edges in $H$.
  Also, a high-degree cluster has at least $2d$ incoming edges.

  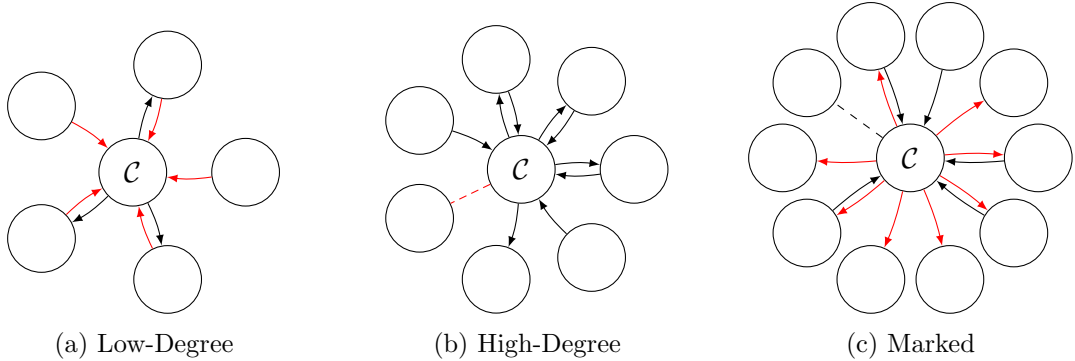
\begin{figure}
    \centering
    \begin{subfigure}[b]{0.3\textwidth}
      \centering
      \begin{tikzpicture}
        \node[draw, circle, minimum size=0.9cm] (center) at (0,0) {$\mathcal{C}$};
        \foreach \a in {1,2,...,5}{
          \node[draw, circle, minimum size=0.9cm] (c-\a) at (\a*360/5: 1.5cm) {};
        }
        \foreach \a in {1,2,3,4,5}{
          \draw[latex-, red] (center) to[bend right=8] (c-\a);
        }
        \foreach \a in {1,3,4}{
          \draw[-latex] (center) to[bend left=8] (c-\a);
        }
      \end{tikzpicture}
      \caption{Low-Degree}
    \end{subfigure}
    \begin{subfigure}[b]{0.3\textwidth}
      \centering
      \begin{tikzpicture}
        \node[draw, circle, minimum size=0.9cm] (center) at (0,0) {$\mathcal{C}$};
        \foreach \a in {1,2,...,7}{
          \node[draw, circle, minimum size=0.9cm] (c-\a) at (\a*360/7: 1.5cm) {};
        }
        \foreach \a in {1,2,3,6,7}{
          \draw[latex-] (center) to[bend right=8] (c-\a);
        }
        \foreach \a in {1,2,5,7}{
          \draw[-latex] (center) to[bend left=8] (c-\a);
        }
        \foreach \a in {4}{
          \draw[densely dashed, red] (center) to (c-\a);
        }
      \end{tikzpicture}
      \caption{High-Degree}
    \end{subfigure}
    \begin{subfigure}[b]{0.3\textwidth}
      \centering
      \begin{tikzpicture}
        \node[draw, circle, minimum size=0.9cm] (center) at (0,0) {$\mathcal{C}$};
        \foreach \a in {1,2,...,10}{
          \node[draw, circle, minimum size=0.9cm] (c-\a) at (\a*360/10: 1.7cm) {};
        }
        \foreach \a in {2,3,6,9,10}{
          \draw[latex-] (center) to[bend right=5] (c-\a);
        }
        \foreach \a in {1,3,5,6,7,8,9,10}{
          \draw[-latex, red] (center) to[bend left=5] (c-\a);
        }
        \foreach \a in {4}{
          \draw[dashed] (center) to (c-\a);
        }
      \end{tikzpicture}
      \caption{Marked}
    \end{subfigure}
    \caption{Different states of a cluster $\mathcal{C}$ in the virtual graph $H$, where clusters are vertices and an edge $\mathcal{C} \to \mathcal{C}'$ means that the center of $\mathcal{C}'$ received the identifier of $\mathcal{C}$.
    Dashed lines indicate that clusters are neighboring, but neither center received the ID of the other center.}
    \label{fig:clusters}
  \end{figure}

  \paragraph*{Making $H$ undirected with small degrees:}
  One problem is that $H$ is a directed graph with possibly large out-degrees, while we would like to have an undirected graph with small degrees.
  Additionally we would like to keep the fact that all low degree clusters are adjacent to all their neighboring clusters in this virtual graph.
  For that, we first \emph{mark} clusters of extremely high out-degree as follows:
  we reverse the communication direction of the previous paragraph, but instead of sending just one message per round along each edge, we send up to $4d^2$ messages.
  This increases the number of rounds by at most a $4d^2$ factor.
  This way, if a message from some cluster $\mathcal{C}$ was sent along an edge in the previous phase, we send up to $4d^2$ messages in the opposite direction, all from clusters that received the identifier of $\mathcal{C}$.
  If more clusters received the identifier of $\mathcal{C}$, we just inform $\mathcal{C}$ that it will be marked.
  This can be done within $O(k \cdot d^3) = k \cdot 2^{O(\sqrt{\log N})}$ rounds, as every round from the previous paragraph now takes $4d^2$ as long.
  Also, at most $(2d+1)\cdot 4d^2 \leq 12 d^3$ many messages are sent along each edge.
  As in the previous paragraph, we can now convergecast the identifiers of at most $4d^2$ outgoing neighbors in $O((d^2 + k \cdot h_i) \cdot 2^{O(\sqrt{\log N})}) = k \cdot2^{O(\sqrt{\log N})}$ rounds to the cluster centers, marking them the same way as before.

  Now, we temporarily remove marked clusters from $H$; we later discuss how to deal with them.
  Note that there are at most $\frac{n}{2d^{i+1}}$ many marked clusters.
  This is because each cluster has in-degree at most $2d$, which means at most a $1/(2d)$ fraction of clusters can have out-degree exceeding $4d^2$.
  This is at most $\frac{n}{d^i} \cdot \frac{1}{2d}$ many clusters, by invariant (A).
  We now have an undirected virtual graph on the clusters, which has degree at most $4d^2$.

  \paragraph*{Computing a Maximal 2-Independent Set in H:}
  $H$ has now degree at most $4d^2$, but we need an additional fact to ensure that we can simulate the communication along $H$ in $G$:
  Every edge is part of at most $12 d^3 = 2^{O(\sqrt{\log N})}$ edges of $H$.
  This is because we can think of every message in the previous phase as trying to establish an edge between two clusters $\mathcal{C}$,  $\mathcal{C}'$ in $H$.
  Such an edge is only established if a message from $\mathcal{C}$ actually reaches $\mathcal{C}'$.
  As every edge in $G$ only forwarded $12 d^3$ many such messages, it can only be part of as many edges in $H$.

  This means that we can now simulate one round of {\sffamily CONGEST} model on $H$ in $O(k \cdot d^3 + (d^2 + k \cdot h_i) \cdot 2^{O(\sqrt{\log n})}) = k \cdot d^3 \cdot 2^{O(\sqrt{\log N})}$ rounds of the base graph.
  This is because every edge is additionally only part of at most $2^{O(\sqrt{\log N})}$ clusters, by our invariant (C).
  Using that, we first compute a coloring of $H^2$, hence ensuring that any two clusters that are within 2 hops in $H$ have different colors.
  That can be done in $O(d^{8} \log^* S \cdot (k \cdot d^3 \cdot 2^{O(\sqrt{\log n})}) = O(d^{11} \log^* S \cdot 2^{O(\sqrt{\log N})})$ rounds, using Linial's algorithm \cite{linial1987distributive}, which runs in $O(\Delta_H^2 \log^* S)$ rounds, as $H^2$ has maximum degree $\Delta_H = O(d^4)$.
  Then, we compute a maximal 2-independent set $C^*$ of high-degree clusters (this is the definition of high degree mentioned above, which is with respect to the neighborhood of clusters in $G^k$).
  Here, 2-independent set means that no two clusters in $C^*$ should share a common neighboring cluster in $H$.
  We can do so by going through all colors one by one, adding clusters to $C^*$ that do not already have a cluster from $C^*$ within distance two in $H$.
  Any $\mathcal{C}$ that has a $\mathcal{C}' \in C^*$ within its 2-cluster-hops joins the new cluster being formed at the center of $\mathcal{C}'$.
  As all high-degree clusters not in $C^*$ must have a neighbor in $C^*$ within 2-cluster hops, all high-degree clusters are part of a newly formed cluster.

  \paragraph*{Forming new clusters:}
  Each high-degree cluster $\mathcal{C}' \in C^*$ has two cases: (I) either none of the neighboring clusters of $\mathcal{C}'$ were marked, in which case all of them will join the new cluster being formed by $\mathcal{C}'$.
  This means that the new cluster contains at least $2d$ many old clusters.
  Thus, there are at most $\frac{n}{d^i}\cdot \frac{1}{2d}$ many such new clusters.
  (II) at least one of the neighboring clusters of $\mathcal{C}'$ was marked.
  In this case, after $\mathcal{C}'$ accepts the clusters that want to join with it, $\mathcal{C}'$ picks one of its marked neighbors and joins a new cluster centered at that marked cluster.
  To make clusters learn about neighbors, use $k$ rounds of flooding, initiated at all nodes of marked clusters.
  This way, we have to send the identifier of at most one marked cluster along each edge, to ensure that all clusters know if they have a marked neighbor.
  Since there are at most $\frac{n}{2d^{i+1}}$ many marked clusters, the number of the new clusters of this kind is also at most $\frac{n}{2d^{i+1}}$.

  \paragraph{Proving the inductive invariants:}
  By the previous paragraph, we have at most $\frac{n}{d^{i+1}}$ many new clusters, proving invariant (A).
  Regarding invariant (B), first notice that each new cluster that we form is made of some of the previous clusters, all of which were within $O(1)$ cluster hops (w.r.t. distances in $G^k$) of the center of the merge (either in $C^*$ or a marked cluster).
  Hence, the maximum cluster radius grows by at most a factor of $O(1)$, which shows that each cluster radius in phase $i$ is at most $(O(1))^{i+1}$ in $G^k$.

  For invariant (C), we have already argued that due to merges between non marked clusters, every edge is used by at most $12 d^3$ many additional clusters, as these merges only happen along edges of $H$.
  For the merging centered at a marked cluster $\mathcal{C}$, we have that if an edge $e$ is part of a path that informed some other clusters about $\mathcal{C}$, they might merge with $\mathcal{C}$ or some other marked cluster.
  In either case, $e$ is included in at most one additional cluster.
  As all of those will merge to the same cluster, we have that $e$ is used by at most $12 d^3 + 1 \leq 13 d^3$ additional clusters.
  By induction, there are at most $i \cdot 13 d^3 + 13d^3 = (i+1) \cdot 13d^3$ many spanning trees that include a given edge.

  \paragraph*{Coloring low-degree old clusters that remain:}
  Finally, we are left with only low-degree clusters, as we have included all high-degree clusters in a new cluster.
  This means that all remaining clusters have at most $2d$ neighboring clusters.
  We can color these cluster using $O(d^2)$ colors by applying Linial's algorithm \cite{linial1987distributive} which runs in $O(\log^* S)$ rounds of {\sffamily CONGEST} on top of the cluster graph, that is, in $2^{O(\sqrt{\log N})}\log^* S$ rounds of the base graph.
  As we use different colors for each phase, we get a total of $\sqrt{\log N} \cdot O(d)^2 = 2^{O(\sqrt{\log N})}$ colors.
\end{proof}

\begin{remark} \label[remark]{rem:netdecompoverlap}
  Even though edges can be part of up to $2^{O(\sqrt{\log N})}$ many clusters, per color class they can be included in at most one cluster.
  This is because otherwise we would have two clusters with the same color that are at distance less than $k$.
\end{remark}

\section{Implications on MIS}
\label{sec:mis}
\vspace{-5pt}

In this section we present our improved algorithm for computing maximal independent set in the {\sffamily CONGEST} model.
In particular, we prove the following:
\begin{theorem} \label{thm:congestMIS}
  There is a randomized distributed algorithm, with $O(\log n)$-bit messages, that computes an MIS in $O(\log \Delta \cdot \sqrt{\log \log n}) + 2^{O(\sqrt{\log \log n})}$ rounds, w.h.p.
\end{theorem}

We will use the following results about Ghaffari's algorithm for computing a (maximal) independent set \cite{ghaffari2016improved}.
To make the paper self-contained, we include a pseudo-code of it in \Cref{app:ghaffarimis}.
\begin{theorem}[ \cite{ghaffari2016improved} ] \label{thm:ghaffariMIS}
  For each node $v$, the probability that $v$ has not made its decision within the first $O(\log \deg(v) + \log 1/\epsilon)$ rounds, where $\deg(v)$ denotes $v$'s degree at the start of the algorithm, is at most $\epsilon$.
\end{theorem}

\begin{lemma}[ \cite{ghaffari2016improved} ] \label[lemma]{lem:shatterMIS}
  Let $B$ be the set of nodes remaining undecided after $\Theta(\log \Delta)$ rounds.
  Then, with high probability, we have the following properties:
  \begin{enumerate}[(P1)]
    \item There is no $(G^{4})$-independent $(G^9)$-connected subset $S \subseteq B$ s.t. $|S| \geq \log_{\Delta} n$.
      This means that $S$ is an independent set in $G^4$ and induces a connected subgraph in $G^9$.
    \item All connected components of $G[B]$, that is the subgraph of $G$ induced by nodes in $B$, have each at most $O(\log_{\Delta} n \cdot \Delta^4)$ nodes.
  \end{enumerate}
\end{lemma}
The statement of \Cref{lem:shatterMIS} is known as a \emph{shattering} guarantee, which is used in various (distributed) algorithms, see e.g. \cite{beck1991algorithmic, barenboim2016locality, alon2012space}.
Intuitively, this means that after $O(\log \Delta)$ rounds of the algorithm, the components induced by undecided nodes are ``small'', or more precisely in this case: they do not contain a large 5-independent set.
If we allowed for messages of unbounded size, we could just think of the remaining components as graphs of size $O(\log n)$, and use traditional algorithms to solve the problem.
However, as we restrict messages to $O(\log n)$-bits, we will need some additional ideas.

We will also use the following ruling set algorithm of Ghaffari\cite{ghaffari2019distributed}:
\begin{lemma}[ \cite{ghaffari2019distributed} ] \label[lemma]{lem:RSghaffari}
  There is a randomized distributed algorithm in the $\mathsf{CONGEST}$ model that, in any network $H = (V,E)$ with at most $n$ vertices, and for any $B' \subset V$ and for any integer $k \geq 1$, with high probability, computes a $(k, 10 k^2 \log \log n)$ ruling set $B^* \subseteq B'$, with respect to distances in $H$, in $O(k^2 \log \log n)$ rounds.
\end{lemma}

\paragraph{Algorithm Outline:}
Combining these results, we can obtain the following:
First, we run the algorithm of Ghaffari\cite{ghaffari2016improved} for $O(\log \Delta)$ rounds, which results in a state described by \Cref{lem:shatterMIS}.
Then, we compute a $(5, O(\log \log n)$ ruling set of each remaining component, using \Cref{lem:RSghaffari}.
This ruling set induces a clustering, where each node is vertex is clustered to its closest node from the ruling set.
By property (P1) of \Cref{lem:shatterMIS}, this yields $N = O(\log n)$ clusters per component of the remaining graph.
Let us call one such cluster a \emph{meta-node}, and note that it has diameter $r = O(\log \log n)$.
For the remainder of this section, let $H$ be the graph, where the vertex set are all meta-nodes and where two clusters are connected if they contain two adjacent nodes.
Then, we compute a network decomposition of $H$ into \emph{super-clusters}.
Going through the color classes of this decomposition, one by one, we compute an MIS of each super-cluster.
We use the fact that these are graphs of low-diameter to amplify the success probability of a randomized algorithm.

The main challenge will be to compute a suitable network decomposition of $H$.
In particular, we aim to compute a network decomposition using as few colors as possible.
In \Cref{lem:impnetdecomp} we obtain such a decomposition, enabling us to prove \Cref{thm:congestMISslow}.
We strengthen this result in \Cref{lem:netdecompmis}, using a randomized approach, to get a decomposition that will be sufficient to prove \Cref{thm:congestMIS}.

\subsection{First Approach: Slower but Simpler}
\label{sec:misslow}

In this section, we prove the following Theorem.
While it give a slower runtime than \Cref{thm:congestMIS}, it is still faster than previous algorithms.
\begin{theorem} \label{thm:congestMISslow}
  There is a randomized distributed algorithm, with $O(\log n)$-bit messages, that computes an MIS in $O(\log \Delta \cdot \log \log n) + 2^{O(\sqrt{\log \log n})}$ rounds, w.h.p.
\end{theorem}

To prove \Cref{thm:congestMISslow}, we use the following algorithm for network decomposition:

\begin{lemma} \label[lemma]{lem:impnetdecomp}
  Let $H$ be the $N = O(\log n)$-node graph as described in the outline.
  There is a deterministic distributed algorithm, with $O(\log n)$-bit messages, that computes a $(O(\log \log n), 2^{O(\sqrt{\log \log n})})$ network decomposition of $H$ into \emph{super-clusters} in $2^{O(\sqrt{\log \log n})}$ rounds of communication on $G$.
\end{lemma}

\begin{proof}
  We will first argue that we can compute a network decomposition of $H^k$.
  Then we refine this decomposition into a decomposition of $H$, while reducing the number of colors used.
  We will do so by using a \emph{ball growing} process, inspired by \cite{awerbuch1985complexity, awerbuch1990sparse, linial1993low}:
  Here, a ball is just a set of vertices with low diameter.
  Starting from balls being clusters of one color, we grow each of them hop by hop in $H$ until it contains enough meta-nodes.
  We use the fact that we have a decomposition of $H^K$ (for $K$ large enough) to argue that different clusters can operate independently.

  \paragraph{Intermediate Network Decomposition:}
  First, we will compute a network decomposition of $H^K$ for $K = \Theta(\log \log n)$.
  In $G$, every meta-node of $H$ is a cluster of diameter $O(\log \log n)$, so we compute such a network decomposition of $H^K$ by computing a network decomposition of $G^k$ for $k = \Theta((\log \log n)^2)$:
  Using the initial partition as a starting point, we get a $(2^{O(\sqrt{\log \log n})}, 2^{O(\sqrt{\log \log n})})$ network decomposition of $G^k$ by \Cref{rem:netdecomp}.
  Note that by design of the algorithm, all nodes of such an initial cluster will end up in the same cluster as well.
  As these initial clusters have diameter $O(\log \log n)$ and we set $k = \Theta((\log \log n)^2)$, two clusters are at distance $\Theta(\log \log n)$ in $H$.

  Now we have an intermediate $(2^{O(\sqrt{\log N})}, 2^{O(\sqrt{\log N})})$ network decomposition of $H^K$.
  That is, every two meta-nodes from different clusters of the same color have distance at least $K = O(\log N)$ in $H$.
  We can simulate one round of communication within clusters of $H$ in $2^{O(\sqrt{\log n})}$ rounds in $G$.

  \paragraph{One Step of Ball Growing:}
  The next step is to refine this intermediate decomposition to compute a new decomposition of $H$ with the properties from \Cref{lem:impnetdecomp}.
  To do so we use the following ball growing process:
  In each step, we add some meta-nodes to a new super-cluster, while deactivating another set of meta-nodes.
  Initially, all meta-nodes are active.

  More precisely, the $i^{th}$ step is as follows:
  Starting from all clusters of color $i$, we initiate a ball growing process.
  Note that we only consider meta-nodes that are still active and not yet part of a super cluster.
  We call a meta-node of $H$ a \emph{boundary} for this ball if at least one of its neighbors is in a different ball.
  We call a ball \emph{good} if there are less boundary than non-boundary noes.

  Initially, a ball is a cluster of color $i$ (or rather its remaining meta-nodes).
  If the ball is not good, we grow it by one hop in $H$.
  We can do this along edges of $H$, which are also edges in $G$.
  In that case, by definition of a good ball, this ball grows by at least a 2 factor in terms of its number of meta-nodes.
  We repeat this until we reach a good ball.
  That happens within $\log N$ steps of growth as otherwise the ball would have more than $2^{\log N} = N$ meta-nodes of $G$, which is not possible.
  Notice that each step of growth can be performed in $2^{O(\sqrt{\log N})}$ rounds:
  We aggregate the number of boundary and non-boundary meta-nodes at the center, which then decides whether to stop or continue the process.
  Once a ball is good, we deactivate its boundary meta-nodes for this phase.
  The non-boundary meta-nodes of each ball are joined together as one super-cluster of the output-decomposition.
  In each step of the ball growing, the radius of these super-clusters increases by at most one and thus stays $2^{O(\sqrt{\log N})}$ (which is also true in $G$, as every meta-node has radius $O(\log N)$).
  Additionally, balls can grow along each edge at most once, meaning that every edge gets included in at most one additional super-cluster on top of the previous clusters it was included in.
  Together with the diameter staying $2^{O(\sqrt{\log N})}$, this ensures property (B).
  We note that the balls that start from different clusters of color $i$ in the intermediate network decomposition can grow simultaneously.
  They will never reach each other, as originally they were separated by at least $\Omega(\log N)$ hops in $H$ and each ball grows at most $\log N$ hops.

  \paragraph{The Full Algorithm:}
  We perform $\log N$ phases:
  In each phase, we perform $2^{O(\sqrt{\log N})}$ steps of ball growing, one step for each color class of the intermediate network decomposition.
  Once a phase is finished, we reactive all unclustered meta-nodes and move on to the next phase.
  Notice that in each phase, at least half of the remaining meta-nodes join a new super-cluster:
  we only deactivate boundary-nodes and further only do so, whenever we add at least as many nodes to a new super cluster.
  Thus after $\log N$ phases, the graph must be empty.
  In total, we spend $\log N \cdot 2^{O(\sqrt{\log N})} \cdot 2^{O(\sqrt{\log N})} = 2^{O(\sqrt{\log N})}$ rounds.
  As we always deactivate the boundary nodes, super clusters are non-adjacent, which shows property (A).
  For the number of colors, we use only one color per phase, and as there are $\log N = O(\log \log n)$ phases, we use as many colors.
\end{proof}

We can now use this decomposition of $G$, to compute a maximal independent set:

\begin{proof}[Proof of \Cref{thm:congestMISslow}]
  As a first step, we compute $H$ as described before, by running Ghaffari's algorithm \cite{ghaffari2016improved} for $O(\log \Delta)$ rounds, and computing a clustering in the remaining parts of the graph.
  Then, we find a $(O(\log \log n), 2^{O(\sqrt{\log \log n})})$ decomposition of $H$, using \Cref{lem:impnetdecomp}.

  For computing the MIS we proceed as follows:
  We work through each of the $O(\log \log n)$ color classes, spending $O(\log \Delta) + 2^{O(\sqrt{\log \log n})}$ rounds per color class.
  In one color class, we can find one MIS per super cluster, as super clusters of the same color are non-adjacent.

  In every step, all nodes of the active super clusters execute $O(\log n)$ parallel executions of the algorithm of Ghaffari\cite{ghaffari2016improved}, as reviewed in \Cref{thm:ghaffariMIS}.  This can be done without any overhead, as every single execution only uses one-bit messages.
  This super cluster contains $O(\Delta^4 \log n)$ regular nodes, by property (P1) of \Cref{lem:shatterMIS}.
  Running this algorithm for $O(\log (\Delta^4 \log n)) = O(\log \Delta + \log \log n)$ rounds, we find an MIS with probability at least $1 - 1/\mathrm{poly}(\Delta^4 \log n)$. Since all $O(\log n)$ parallel executions are independent, the probability that none of them succeeds is at most $1/\mathrm{poly}(n)$.

  Now we just need to find a run that was successful.
  For this we use the network decomposition we obtained from \Cref{lem:impnetdecomp}.
  First, each node $v$ performs a local check for all runs, by making sure that either $v$ is in the MIS and none of its neighbors is, or that $v$ is not in the MIS, but at least one of its neighbors is.
  This can again be done with just one-bit messages.
  Then, we can convergecast these local checks towards the cluster centers in $2^{O(\sqrt{\log \log n})}$ rounds.
  These centers can pick the first successful run and inform all nodes of their cluster in $2^{O(\sqrt{\log \log n})}$ rounds.
  We remove all nodes that are in the MIS of the super cluster, together with all their neighbors (in the base graph).
  After this, we move on to the next color.

  In total, we spend $O(\log \log n) \cdot (O(\log \Delta + \log \log n) + 2^{O(\sqrt{\log \log n})}) = O(\log \Delta \cdot \log \log n) + 2^{O(\sqrt{\log \log n})})$ rounds and find an MIS with high probability.
\end{proof}

\subsection{Second Approach: Faster}

To improve the runtime compared to \Cref{thm:congestMISslow} we need to obtain a network decomposition using fewer colors.
Instead of a sequential ball growing process, we will perform a randomized \emph{ball carving}, similar to Elkin and Neiman\cite{elkin2016distributed}.
As this decomposition will be used on small components of the graph, we need to ensure that we still succeed with probability $1 - 1/\text{poly}(n)$ even on components with much less than $n$ vertices.
We will use a similar idea as in the proof of \Cref{thm:congestMISslow}, namely that we run many random processes in parallel and use a previously computed network decomposition to find a run that was successful.
However, defining the right measure of success and identifying a successful run both require much more care than in algorithm for MIS.
The resulting algorithm is formalized in the following Lemma:
\begin{lemma} \label[lemma]{lem:netdecompmis}
  Let $H$ be the $N = O(\log n)$-node graph as described in the algorithm outline.
  There is a randomized distributed algorithm, with $O(\log n)$ bit messages, that computes a strong diameter $(O(\sqrt{\log \log n}), 2^{O(\sqrt{\log \log n})})$ network decomposition of $H$ in $2^{O(\sqrt{\log \log n})}$ rounds, with probability $1 - 1/\mathrm{poly}(n)$.
\end{lemma}

\subsubsection{Reviewing Ball Carving}
\label{sec:ballcarving}

Before giving a proof of \Cref{lem:netdecompmis}, we review previous approaches to ball carving:
we use the same process as Elkin and Neiman\cite{elkin2016distributed} used for obtaining a strong diameter $(O(\log N), O(\log N))$ network decomposition.
Their algorithm is based on the \emph{shifted shortest path approach}, due to Blelloch et al.\cite{blelloch2014nearly} and Miller et al.\cite{miller2013parallel}.

In this approach, one step of ball carving works as follows:
We will remove one \emph{block} $W \subset V$ from the graph.
A connected component of $G[W]$ is a \emph{cluster}.
Each vertex $v$ picks a random value $r_v$ from the exponential distribution with parameter $\beta$, denoted $\mathcal{EXP}(\beta)$, which as density:
\[
  f(x) =
  \begin{cases}
    \beta \cdot e^{- \beta x} & x \geq 0 \\
    0                         & \text{otherwise}.
  \end{cases}
\]
For analysis, suppose every node $v$ broadcasts the value $r_v$ to all other nodes within distance $R_v = \lfloor r_v \rfloor$.
Each node $u$ keeps track of the values $m_i = r_{v_i} - d(u, v_i)$ it receives and orders them non-increasingly (i.e. $m_1 \geq m_2 \geq \dots$).
If $m_1 - m_2 > 1$, we add $u$ to $W$ and say that $u$ joins the cluster centered at $v_1$.
In the case that $m_1 - m_2 \leq 1$, the node $u$ is not clustered in this round.
Note that for nodes to make this decision, it suffices to forward the largest two values $m_i, m_j$ they have received so far.

In \cite{elkin2016distributed} the authors show that this clustering has two properties:
The clusters are non-adjacent and if $R = \max_v R_v$, clusters have strong diameter at most $2R - 2$.
Moreover, as nodes only need to forward the two largest values $m_i, m_j$ in each round, $O(\log R)$ bit messages suffice.
Assuming $R = \text{poly}(N)$, this ensures that the algorithm also works in the $\mathsf{CONGEST}$ model.
To get a network decomposition, they argue that after $O(\log N)$ rounds of ball carving, the graph is exhausted and that there is never a value $r_v$ larger than $O(\log N)$ (both with high probability).

As we will want to compute a network decomposition of graphs with just $N = O(\log n)$ nodes, this success probability of $1 - 1/\text{poly}(N) = 1 - 1/\text{poly}(\log n)$ will not be high enough.
We will amplify this probability of success by performing multiple runs of ball carving in parallel.
The network decomposition from \Cref{thm:netdecomp} will then be used to determine a successful run, while the fact that two clusters of the same color have distance at least $k$, will allow the clusters to run independently.

We will also use the following property of the exponential distribution, proven by Miller et al.\cite{miller2013parallel}:
\begin{lemma} [ \cite{miller2013parallel} ]
  \label[lemma]{lem:expdist}
  Let $d_1 \leq \dots \leq d_n$ be arbitrary values and $\delta_1, \dots, \delta_n$ be independent random variables picked from $\mathcal{EXP}(\beta)$.
  Then the probability that the largest and the second largest value of $\delta_j - d_j$ are within 1 of each other is at most $\beta$.
\end{lemma}
Applied to our ball carving, this means that a node is added to a cluster with probability at least $\beta$.

\subsubsection{Proof of \Cref{lem:netdecompmis}}

Using the tools introduced in the previous section, we are now ready to prove \Cref{lem:netdecompmis}.
\begin{proof}[Proof of \Cref{lem:netdecompmis}]
  We will show how to obtain a $(O(\sqrt{\log \log n}), 2^{O(\sqrt{\log \log n})})$ network decomposition of $H$ from a $(2^{O(\sqrt{\log \log n})}, 2^{O(\sqrt{\log \log n})})$ decomposition of $H^k$.

  \paragraph{Proof Outline:}
  We restart from scratch, using the network decomposition only for amplifying success probabilities.
  We perform a number of phases, where in each phase we perform multiple steps of ball carving, initiating the process from nodes of one color class.
  As a single execution of ball carving does not have high enough success probability, we run a number of independent executions in parallel and pick one that was successful using the initial network decomposition.
  In this step, we will make use of the fact that we have a decomposition of $H^k$, which prevents different clusters of the same color to overlap in their ball carving, making it possible to pick one successful run per cluster.

  \paragraph{Initial Network Decomposition:}
  Using the clustering given by the meta-nodes as an initial clustering, we can compute a $(2^{O(\sqrt{\log \log n})}, 2^{O(\sqrt{\log \log n})})$ decomposition of $G^k$ for $k = \Omega(\log \log n \cdot 2^{2\sqrt{\log N}})$, by \Cref{thm:netdecomp}.
  As the algorithm of \Cref{thm:netdecomp} does not break initial clusters, this is also a network decomposition of $H^K$ for $K = \Omega(2^{2\sqrt{\log N}})$, since meta-nodes have diameter $O(\log \log n)$.

  \paragraph{Algorithm outline:}
  Next, we want to transform this $(2^{O(\sqrt{\log N})}, 2^{O(\sqrt{\log N})})$ decomposition of $H^K$ to a decomposition of $H$, while reducing the number of colors.
  Note that we can simulate one round of communication on $H$ in $O(\log \log n)$ rounds of communication on $G$, as long as every node only needs to learn about a constant number of messages from its (potentially up to $O(\log n)$ many) neighbors.
  As mentioned, the ball carving described in \Cref{sec:ballcarving} has this property.
  From now on, let a round denote one such round of communication on $H$.

  The algorithm consists of phases $t = 1, \dots \sqrt{\log N}$.
  Let $H_1 = H$.
  In each phase, we remove a \emph{block} $W_t$ from the current graph and set $H_{t+1} = H_t \setminus W_t$.
  The $t$-th phase is as follows:
  We have $2^{O(\sqrt{\log N})}$ iterations of $2^{O(\sqrt{\log N})}$ rounds each.
  In the $i^{th}$ iteration, we perform one step of ball carving:
  Let $\mathcal{C}_i \subseteq H_t$ denote all remaining and active meta-nodes of color $i$, let $\beta = 2^{- \sqrt{\log N} - 2}$, and let $d = 2^{2 \sqrt{\log N}}$.
  At the start of a phase, all meta-nodes are active.
  Now we perform one step of ball carving as described above, but with the key-difference that only meta-nodes $v \in \mathcal{C}_i$ pick a random value $r_v \sim \mathcal{EXP}(\beta)$.
  In the following broadcast, all meta-nodes from $H_t$ that are still active, participate and behave the same as in the regular ball carving.
  Note that some vertices might not be reached by this process, we just ignore them for this iteration.
  Now, we add all meta-nodes that are clustered by this process to $W_t$ and all meta-nodes that received some value, but did not get clustered, are set to be \emph{inactive} for this phase.
  As shown in \cite{elkin2016distributed}, this yields non-adjacent clusters of strong diameter at most $2 \max_v r_v$.

  \paragraph{Analysis:}
  We call a run of the algorithm \emph{successful} if the following holds:
  A $2^{-\sqrt{\log N}}$ fraction of meta-nodes that received any value $r_v$ is in $W_t$, and the maximum value of $r_v$ chosen by the meta-nodes does not exceed $d$.
  The second part is true with probability at least $3/4$, as for $X \sim \mathcal{EXP}(\beta)$ we have $\Pr[X \geq d + 1] \leq e^{-\beta(d+1)}$.
  By our choice of parameters:
  \begin{align*}
    \Pr[X \geq d + 1] & \leq \exp\left(- 2^{-\sqrt{\log N} - 2} \cdot \left(2^{2 \sqrt{\log N}} + 1\right)\right) \\
                      & \leq \exp\left(- 2^{\sqrt{\log N} - 2} \right) \ll \frac{1}{4N},
  \end{align*}
  where the last inequality hold for $N$ large enough.
  Thus, by a union bound over all $v \in H_t$, the maximum $r_v$ is at most $d$ with probability at least $1 - 1/4 = 3/4$.

  For the number of meta-nodes that are deactivated, we use the following claim:
  \begin{claim}
    With probability at least $3/4$ a $2^{- \sqrt{\log N}}$ fractions of all meta-nodes in $H_t$ that received a message from a node $v \in \mathcal{C}_t$ are added to $W_t$.
  \end{claim}
  \begin{proof}
    Let $N_t$ denote the number of remaining and active meta-nodes that received some message.
    Using \Cref{lem:expdist} and setting $d_i = d(u, v_i)$ and $\delta_i = r_{v_i}$, we have that the probability that $m_1 - m_2 \leq 1$, is at most $\beta$.
    This means that the probability of a meta-node from $N_t$ not being included in $W_t$ is at most $\beta$ as well.
    Letting $X_t$ denote the number of unclustered meta-nodes from $N_t$, and using linearity of expectation, we get $\E [X] \leq \beta N_t$.
    We can now apply Markov's inequality:
    \begin{align*}
      \Pr\left[X \geq 2^{- \sqrt{\log N}} N_t \right] \leq \frac{\beta N_t}{2^{-\sqrt{\log N}}} = \frac{2^{-\sqrt{\log N} - 2}}{2^{-\sqrt{\log N}}} = 1/4
    \end{align*}
    Thus, with probability at least $3/4$ the desired event holds.
  \end{proof}
  Combining these results, we get that a single run is successful with probability at least $1/2$.
  Note that this algorithm can be implemented using $O(\log d)$-bit messages:
  if a meta-node chooses a larger random value than $d$, we can abort this execution.
  As such a run is unsuccessful by definition, this modification does not change our success probability.
  Additionally, the round complexity is $O(d) = 2^{O(\sqrt{\log \log n})}$.

  As we can send $O(\log n)$ bit messages per round and $O(\log d) = O(\sqrt{\log \log n})$, we can execute $O(\log n)$ many parallel executions of the ball carving in parallel, while losing just a $O(\sqrt{\log \log n})$ factor overhead.
  That is, within $O(\sqrt{\log \log n}) \cdot 2^{O(\sqrt{\log \log n})} = 2^{O(\sqrt{\log \log n})}$ rounds, we can execute $O(\log n)$ runs of ball carving.
  As an individual run has success probability $1/2$, the probability that all runs fail is at most $1/\text{poly}(n)$.

  \paragraph{Identifying a successful run:}
  First, note that we can have different clusters operate independently, as they are separated by $\Omega(2^{2\sqrt{\log N}})$ hops, while the balls have radius at most $d = 2^{\sqrt{\log N}}$.
  As noted before, if a meta-node $v$ chooses a value $r_v \geq d + 1$, we mark the run as failed.
  For a single run and a single cluster, we can now aggregate the number of clustered and unclustered meta-nodes at the cluster center, by a simple convergecast.
  The diameter of the meta-nodes that initiated the ball carving is at most $2^{O(\sqrt{\log \log n})}$, by the parameters of the initial network decomposition.
  Additionally, meta-nodes at distance up to $d$ might have joined a newly-formed cluster.
  Thus in $d + 2^{O(\sqrt{\log \log n})}  = 2^{O(\sqrt{\log \log n})}$ rounds, we can determined if a run was successful at the center meta-node of the cluster.

  For all runs, we just need a $O(\log \log n)$ round overhead:
  As there are at most $N = O(\log n)$ meta-nodes in total, $O(\log \log n)$ bit suffice to describe these numbers.
  Again, as we have access to $O(\log n)$ bit messages, we can execute $O(\log n)$ runs in parallel with a $O(\log \log n)$ factor overhead.
  Finally, the center meta-node picks the first successful run and announces it to all nodes that participated in some execution.

  This describes one iteration.
  By construction, at the end of a phase, all meta-nodes are either in a cluster or deactivated.
  Furthermore, at least a $2^{\sqrt{\log N}}$ fraction of meta-nodes joined a cluster.
  Thus, after $\sqrt{\log N}$ phases there are no meta-nodes remaining.
  For every phase, we can use a single color:
  As shown in \cite{elkin2016distributed}, clusters that are created at the same time are non-adjacent.
  This argument can easily be extended to the case where we have multiple different phases, but deactivate the \emph{boundary} nodes of each cluster (i.e. the meta-nodes that received $m_1 - m_2 \leq 1$.
  This means that we only need $\sqrt{\log \log n}$ colors.
  The diameter of the new clusters is at most $d = 2^{O(\sqrt{\log \log n})}$ by definition of a successful run.
\end{proof}

\subsubsection{Proof of \Cref{thm:netdecomp}}

The proof is analogous to the proof of \Cref{thm:congestMISslow}. Instead of using the network decomposition from \Cref{lem:impnetdecomp}, we use the decomposition obtained in \Cref{lem:netdecompmis} to get a better runtime.
\begin{proof}[Proof of \Cref{thm:congestMIS}]
  As a first step, we run the algorithm of Ghaffari\cite{ghaffari2016improved} for $O(\log \Delta)$ rounds.
  This will yield an independent set of nodes, which we remove from the graph, together with their neighbors.
  By \Cref{lem:shatterMIS} all nodes that remain, form components of size $O(\Delta^4 \log n)$, by property (P2).
  Next, we compute a $(5, O(\log \log n))$ ruling set $B_C^*$, using \Cref{lem:RSghaffari}, for every remaining component $C$.
  This takes $O(\log \log n)$ rounds.
  Now we cluster the nodes around vertices of $B_C^*$, where each node joins the cluster formed by the closest $B_C^*$ vertex, breaking ties by IDs.
  This clustering can be performed in $O(\log \log n)$ rounds as well, by $B_C^*$ nodes broadcasting their identifiers.
  Each node joins the cluster formed by the node who's identifier it receives first, breaking ties arbitrarily.
  By property (P1) of \Cref{lem:shatterMIS}, this yields at most $O(\log n)$ clusters per component.
  Let us call one such cluster a \emph{meta-node}.
  We can now apply \Cref{lem:netdecompmis}, to get a $(O(\sqrt{\log \log n}), 2^{O(\sqrt{\log \log n})})$ decomposition of the meta-nodes in each component.

  For computing the MIS we proceed as follows:
  We work through each of the $O(\sqrt{\log \log n})$ color classes, spending $O(\log \Delta) + 2^{O(\sqrt{\log \log n})}$ rounds per color class.
  In one color class, we can find one MIS per super cluster, as super clusters of the same color are non-adjacent.

  In every step, all nodes of the active super clusters execute the algorithm of Ghaffari\cite{ghaffari2016improved}, as reviewed in \Cref{thm:ghaffariMIS}, running $O(\log n)$ parallel executions of it.  This can be done without any overhead, as every single execution only uses one-bit messages.
  Note that this super cluster contains $O(\Delta^4 \log n)$ regular nodes, by property (P1) of \Cref{lem:shatterMIS}.
  Running this algorithm for $O(\log (\Delta^4 \log n)) = O(\log \Delta + \log \log n)$ rounds, we find an MIS with probability at least $1 - 1/\mathrm{poly}(\Delta^4 \log n)$. Since all $O(\log n)$ parallel executions are independent, the probability that none of them succeeds is at most $1/\mathrm{poly}(n)$.

  Now we just need to find a run that was successful.
  For this we use the network decomposition we obtained from \Cref{lem:impnetdecomp}.
  First, each node $v$ performs a local check for all runs, by making sure that either $v$ is in the MIS and none of its neighbors is, or that $v$ is not in the MIS, but at least one of its neighbors is.
  This can again be done with just one-bit messages.
  Then, we can convergecast these local checks towards the cluster centers in $2^{O(\sqrt{\log \log n})}$ rounds.
  These centers can pick the first successful run and inform all nodes of their cluster in $2^{O(\sqrt{\log \log n})}$ rounds.
  We remove all nodes that are in the MIS of the super cluster, together with all their neighbors (in the base graph).
  After this, we move on to the next color.

  In total, we spend $O(\sqrt{\log \log n}) \cdot (O(\log \Delta + \log \log n) + 2^{O(\sqrt{\log \log n})}) = O(\log \Delta \cdot \sqrt{\log \log n}) + 2^{O(\sqrt{\log \log n})}$ rounds and find an MIS with high probability.
\end{proof}

\newpage

\nocite{elkin2006faster}
\bibliographystyle{alpha}
\bibliography{refs}

\appendix

\section{The Independent Set Algorithm of \cite{ghaffari2016improved}}
\label{app:ghaffarimis}

\begin{mdframed}[hidealllines=false,backgroundcolor=gray!30]
  \vspace{-10pt}
  \paragraph{The Algorithm}
  In each round $t$, each node $v$ has a \emph{desire-level} $p_t(v)$, which is initially set to $p_0(v) = 1/2$.
  We define the \emph{effective degree} $d_t(v)$ of a node $v$ to be the sum of the desire-levels of neighbors of $v$, i.e. $d_t(v) = \sum_{u \in N(v)} p_t(u)$.
  The desire levels change over time as follows:
  \begin{align*}
    p_{t+1} =
    \begin{cases}
      p_t(v)/2,             & \text{if } d_t(v) \geq 2 \\
      \min\{2p_t(v), 1/2\}, & \text{if } d_t(v) < 2.
    \end{cases}
  \end{align*}
  The desire levels are used as follows:
  In each round, node $v$ gets \emph{marked} with probability $p_t(v)$ and if no neighbor of $v$ is marked, $v$ joins the MIS and get removed along with its neighbors.
\end{mdframed}

Note that this algorithm can be implemented using single-bit messages: instead of exchanging the desire levels, it is sufficient to exchange in which of the two way they change.
For indicating whether or not a node is marked or joined the MIS, single bits suffice as well.

\section{Sparse Neighborhood Cover}
\label{app:sparsecover}

In this section, we show how a network decomposition of $G^k$ can be transformed into a sparse neighborhood cover.
Let us first recall its definition, as introduced by Awerbuch and Peleg\cite{awerbuch1990sparse}:
\begin{definition}
  Let $G = (V, E)$ be a graph and let $k \geq 1$, $d \geq 1$, and $s \geq 1$ be three integer parameters.
  A \emph{$s$-sparse $k$-neighborhood cover of diameter $d$} is a collection of clusters $C \subseteq V$ such that
  (A) for each cluster $C$, we have a rooted spanning tree of $G[C]$ of diameter at most $d$,
  (B) each $k$-neighborhood of $G$ is completely contained in some cluster, and
  (C) each node of $G$ is in at most $s$ clusters.
\end{definition}
While there were constructions known in the {\sffamily LOCAL} model \cite{awerbuch1990sparse, awerbuch1996fast}, Ghaffari and Kuhn\cite{ghaffari2018derandomizing} only recently showed that neighborhood covers can also be computed in the {\sffamily CONGEST} model.
In particular, they showed that a $(c,d)$ network decomposition of $G^{2k}$ can be transformed into a $c$-sparse $k$-neighborhood cover of diameter $d + k$.
In contrast to our result, they assumed to be given a strong diameter decomposition of $G^{2k}$, while \Cref{thm:netdecomp} only guarantees weak diameter.
In the next Corollary we show that the construction of \cite{ghaffari2018derandomizing} is also applicable if we are only given a weak diameter decomposition.
\begin{corollary}
  \label[corollary]{cor:sparsecover}
  Assume that we are given a $(c, d)$-decomposition of $G^{2k}$.
  One can compute a $c$-sparse $k$-neighborhood cover of diameter $d + k$ in $O(c(d + k))$ rounds in the $\mathsf{CONGEST}$ model on $G$.
  Consequently, for every $k \geq 1$, one can deterministically compute a $2^{O(\sqrt{\log n})}$-sparse $k$-neighborhood cover of diameter $k \cdot 2^{O(\sqrt{\log n})}$ of and $n$-node graph $G$ in $k \cdot 2^{O(\sqrt{\log n})}$ rounds of $\mathsf{CONGEST}$.
\end{corollary}
\begin{proof}
  We process all colors one by one:
  For each cluster $C$ of the current color, we create a new, larger cluster $C'$, by including $C$ and all nodes of $G$ that are within distance $k$ of some node from $C$.
  First, note that this way all nodes of the spanning tree of $C$ will be in $C'$.
  This is because every path in the spanning tree between two nodes of $C$, that has no internal nodes in $C$, has length at most $2k$ by definition of the network decomposition.
  As both of these endpoints will include their entire $k$ hop neighborhood, the path will be contained in $C'$.

  If two clusters $C_1$ and $C_2$ had the same color in the network decomposition, their resulting clusters of the neighborhood cover will be disjoint, as they were separated by at least $2k$ hops.
  Thus, for every color, a node is in at most one cluster of the computed neighborhood cover.
  The cover is therefore $c$-sparse.
  By construction, the clusters have diameter at most $d+k$ and for every node there is some cluster containing its entire $k$-hop neighborhood.
  For a single color, the new clusters can be computed in $O(d + k)$ rounds, thus the total computation takes $O(c(d + k))$ rounds.

  This proves the first part of the statement.
  The second part now follows by applying \Cref{thm:netdecomp}.
\end{proof}

\subsection{Applications to MST}
\label{app:mst}

Another application of sparse neighborhood covers is in the distributed construction of minimum spanning trees (MST).
Elkin\cite{elkin2004faster} showed that there is a randomized distributed algorithm that computes an MST in $\tilde{O}(\mu(G,\omega) + \sqrt{n})$ of the {\sffamily CONGEST} model, where $\mu(G, \omega)$ is the \emph{MST-radius} of the graph $G$.
The MST-radius $\mu(G, \omega)$ is defined as the minimum $t$ such that every edge not belonging to the MST is the heaviest edge in some cycle of length at most $t$.
The author also remarks that given a deterministic construction of sparse neighborhood covers would yield a similar deterministic algorithm for computing the minimum spanning tree.
While Awerbuch et al.\cite{awerbuch1996fast} gave such constructions using large messages, the only construction applicable in the {\sffamily CONGEST} model was due to Ghaffari and Kuhn\cite{ghaffari2018derandomizing}.
However, they were only able to construct a $2^{O(\sqrt{\log n \log \log n})}$-sparse $k$-neighborhood cover of diameter $k \cdot 2^{O(\sqrt{\log n \log \log n})}$.
Using our improved result from \Cref{cor:sparsecover}, we get the following:
\begin{corollary}
  \label[corollary]{cor:mst}
  There is a deterministic distributed algorithm that computes an MST of an $n$-vertex graph $(G, \omega)$ with distinct edge-weights\footnote{This is without loss of generality, as edge weights can be made distinct using standard techniques.} in $2^{O(\sqrt{\log n})} \cdot O(\mu(G, \omega) + \sqrt{n})$ rounds, where every vertex is given the MST-radius $\mu(G, \omega)$ as input.
  Moreover, this algorithm works in the $\mathsf{CONGEST}$ model.
\end{corollary}
For completeness, we provide a high-level overview\footnote{Elkin\cite{elkin2004faster} used a slightly different approach, allowing for a smooth trade-off between run-time and message size. We present a simpler version, focusing on the $\mathsf{CONGEST}$ model with $O(\log n)$ bit messages.}:
First, we compute a $\mu(G, \omega)$-neighborhood cover, using \Cref{cor:sparsecover}.
This cover will have sparsity $2^{O(\sqrt{\log n})}$ and diameter $\mu(G, \omega) \cdot 2^{O(\sqrt{\log n})}$.
Now we compute an MST of every cluster using a standard MST algorithm, such as the one of Kutten and Peleg\cite{kutten1998fast}.
As clusters have diameter $\mu(G, \omega) \cdot 2^{O(\sqrt{\log n})}$, this algorithm runs in $\mu(G, \omega) \cdot 2^{O(\sqrt{\log n})} + O(\sqrt{n} \cdot \log^* n)$ rounds.
Since we have at most $2^{O(\sqrt{\log n})}$ clusters overlapping, we can compute an MST of each cluster in $2^{O(\sqrt{\log n})} \cdot (\mu(G, \omega) + \sqrt{n})$ rounds in total.
Now the nodes can locally decide if an edge $e$ belongs to the MST of $G$ according to the following rules:
(A) if there is some cluster in which $e$ does not belong to the MST, $e$ does not belong to the MST of $G$.
(B) if $e$ is contained in the MST of every cluster that contains $e$, then $e$ also belong to the MST of $G$.

We still have to argue why these local decisions are correct:
Rule (A) is correct as the edge weights are distinct, which means that the MST is unique.
Rule (B) is correct by the definition of $\mu(G, \omega)$ and the fact that we have a $\mu(G, \omega)$ neighborhood cover:
Every edge $e$ that is not part of the MST is the heaviest edge on a cycle of length at most $\mu(G, \omega)$.
This entire cycle must be included in some cover $C$, which means the MST of $C$ cannot contain $e$.

\section{Spanners and Dominating Set}
\label{app:spanners}

Ghaffari and Kuhn\cite{ghaffari2018derandomizing}, gave deterministic algorithm for computing sparse spanners as well as approximations of dominating set.
One thing they rely on is what they call a \emph{strong diameter $k$-hop $(k f(n), f(n))$-decomposition}, which is a partitioning the network into $f(n)$ colors.
Each color class consists of clusters of strong-diameter $k f(n)$ such that any two nodes from different clusters have distance at least $k$.
However, their applications of this network decomposition do not depend on the fact that the clusters have strong-diameter and are connected in $G$, other than for ensuring efficient communication.

To replace the random parts of various algorithms, the authors of \cite{ghaffari2018derandomizing} define an abstract problem, which they call the \emph{hitting set} problem.
\begin{definition}[The Hitting Set Problem]
  Consider a graph $G = (V, E)$ with two special sets of nodes $L \subseteq V$ and $R \subseteq V$ with the following properties:
  each node $\ell \in L$ knows a set of vertices $R(\ell) \subseteq R$, where $\lvert R(\ell) \rvert = \Theta(p \log n)$, such that $dist_G(\ell,r) \leq T$ for every $r \in R(\ell)$.
  Here, $p$ and $T$ are two given integer parameters in the problem.
  Moreover, there is a $T$-round $\mathsf{CONGEST}$ algorithm that can deliver one message from each node $r \in R$ to all nodes $\ell \in L$ for which $r \in R(\ell)$.
  We emphasize that the same message is delivered to all nodes $\ell \in L$.

  Given this setting, the objective in the hitting set problem is to select a subset $R^* \subseteq R$ such that (I) $R^*$ dominates $L$---i.e., each node $\ell \in L$ has at least one node $r^* \in R$ such that $r^* \in R(\ell)$---and (II) we have $\lvert R^* \rvert \leq \lvert R \rvert / p$.
\end{definition}
\begin{lemma}[\cite{ghaffari2018derandomizing}] \label[lemma]{lem:hitset}
  Given a $2T$-hop $(c,d)$ decomposition of the graph $G$ of the hitting set problem, there is a deterministic distributed algorithm that in $\tilde{O}(c(d + T))$ rounds solves the hitting set problem.
\end{lemma}
\begin{proof}[Proof Sketch]
  The general outline is to derandomize the trivial randomized algorithm which includes every node of $R$ in $R^*$ with probability $1/(2p)$.
  To this end, the authors of \cite{ghaffari2018derandomizing} show that the requirements of the hitting set problem can be reformulated as a cost function depending on the random choices of the nodes from $R$.
  Furthermore, they argue that using just $\Theta(\log n)$-wise independence gives the same concentration as a standard Chernoff bound, up to constant factors.
  Finally, they show how to fix the bits of randomness by using the method of conditional expectation.
  In this last step they use the network decomposition which is assumed to exist.
  However, they only use the fact that clusters have low diameter to compute the conditional expectations and report them to the center of the cluster, which will then fix the bits of randomness.
  Thus, the weak diameter of the decomposition from \Cref{thm:netdecomp} suffices (see also \Cref{rem:netdecompoverlap}).
\end{proof}

\subsection{Spanners}

Having argued that the hitting set computation of \cite{ghaffari2018derandomizing} also works with our weaker form of network decomposition, we directly get the following:
\begin{theorem} \label{thm:spanner}
  There is a distributed deterministic algorithm in the $\mathsf{CONGEST}$ model that computes a $(2k - 1)$-stretch spanner with size $O(kn^{1+1/k} \log n)$ in $2^{O(\sqrt{\log n})}$ rounds.
\end{theorem}
The proof from \cite{ghaffari2018derandomizing} of this Theorem adapts the algorithm of \cite{baswana2007simple} to use the hitting set algorithm instead of a random marking process.
In order to apply the hitting set algorithm, they compute a $(2^{O(\sqrt{\log n \log \log n})}, 2^{O(\sqrt{\log n \log \log n})})$ network decomposition in $2^{O(\sqrt{\log n \log \log n})}$ rounds.
As mentioned, we can replace this network decomposition with the $(2^{O(\sqrt{\log n})}, 2^{O(\sqrt{\log n})})$-decomposition from \Cref{thm:netdecomp}, which can be computed in $2^{O(\sqrt{\log n})}$ rounds.
This immediately implies the Theorem.

\subsection{Minimum Set Cover and Dominating Set}

A set cover instance $(\mathcal{X}, \mathcal{S})$ is given by a set $\mathcal{X}$ of elements and a set $\mathcal{S} \subseteq 2^\mathcal{X}$ of subsets of $\mathcal{X}$ such that $\bigcup_{A \in \mathcal{S}} A = \mathcal{X}$.
The objective of the minimum set cover problem is to select a subset $\mathcal{C} \subseteq \mathcal{S}$ of the sets in $\mathcal{S}$ such that $\bigcup_{A \in \mathcal{C}} = \mathcal{X}$ and the cardinality of $\mathcal{C}$ is minimized.
The set cover instance $(\mathcal{X}, \mathcal{S})$ is modeled as a distributed graph problem by defining a bipartite network graph that has a node $u_{\mathcal{X}}$ for each element $x \in \mathcal{X}$ and a node $v_A$ for every set $A \in \mathcal{S}$ and contains an edge $\{ u_{\mathcal{X}}, v_A \}$ whenever $x \in A$.
Note that with this definition we can model the minimum dominating set problem by having each node $u$ correspond to an element of $\mathcal{X}$ and a set in $\mathcal{S}$ which contains $u$ as well as all its neighbors.

For this distributed variant of set cover, Ghaffari and Kuhn\cite{ghaffari2018derandomizing} showed that a $O(\log^2 n)$ approximation could be computed in $2^{O(\sqrt{\log n \cdot \log \log n})}$ rounds.
As noted, this yields the same result for minimum dominating set.

More recently, Deurer, Kuhn, and Maus\cite{deurer2019deterministic} gave a stronger result for minimum dominating set.
They obtain a $2^{O(\sqrt{\log n \cdot \log \log n})}$ round deterministic distributed algorithm that computes a $O(\log \Delta)$ approximation of minimum dominating set.
They also remark that the time complexity is due to the best known deterministic algorithm for computing a network decomposition of $G^2$ by Kuhn and Ghaffari\cite{ghaffari2018derandomizing} and that improving their result would yield a faster algorithm.
Thus, as a consequence of \Cref{thm:netdecomp} we get:
\begin{theorem}
  There is a deterministic {\sffamily CONGEST} model algorithm that computes a $O(\log \Delta)$-approximation of minimum dominating set in $2^{O(\sqrt{\log n})}$ rounds.
\end{theorem}
While this result only holds for minimum dominating set, it entails an improved result for set cover.
By a standard reduction of set cover to minimum dominating set, one can obtain a similar runtime.
However, this reduction changes the maximum degree of the considered graph, which leads to:
\begin{corollary}
  There is a deterministic {\sffamily CONGEST} model algorithm that computes a $O(\log n)$-approximation of minimum set cover in $2^{O(\sqrt{\log n})}$ rounds.
\end{corollary}

\end{document}